\documentclass{article}
\usepackage{booktabs} 
\usepackage{float}
\usepackage[ruled]{algorithm2e} 

\usepackage[margin=1in]{geometry}
\usepackage{graphicx} 
\usepackage{amsmath, amsfonts, amsthm}
\usepackage{mathrsfs}
\usepackage{biblatex}

\newtheorem{definition}{Definition}
\newtheorem{lemma}{Lemma}
\newtheorem{theorem}{Theorem}
\newtheorem{proposition}{Proposition}

\newcommand{\vr}{\mathcal{R}}
\newcommand{\ati}{A_{-i}^t}
\newcommand{\atti}{A_{-i}^T}
\newcommand{\sat}{\text{sat}}
\newcommand{\score}{\text{score}}

\newcommand{\bw}{\bar{w}}
\newcommand{\argmin}{\text{argmin}}
\newcommand{\argmax}{\text{argmax}}
\newcommand{\flr}[1]{\lfloor #1 \rfloor}
\newcommand{\acup}{A^{\cup}}
\newcommand{\acap}{A^{\cap}}
\newcommand{\strat}{\tilde{A}^t(v_i)}
\newcommand{\tstrat}{\tilde{A}^T(v_i)}
\newcommand{\onstrat}{\tilde{A}^t(v_i)}

\newcommand{\ostrat}{\tilde{A}_t(v_i)}

\newcommand{\onstratdef}[1]{(A_1(v_i), A_2(v_i), \dots, #1)}

\newcommand{\instance}{(N,A^T,C^T)}
\newcommand{\ncap}{N^{\cap}}
\newcommand{\nncap}{\tilde{N}^{\cap}}
\newcommand{\lexmin}{\text{lexmin}}

\author{
Allan Borodin \\
Department of Computer Science \\
University of Toronto
\and
Tristan Lueger \\
Department of Computer Science \\
University of Toronto
}

\date{}

\title{Online Temporal Voting: Strategyproofness, Proportionality and Asymptotic Analysis}

\begin{document}
\begin{titlepage}

\maketitle

\vspace{1cm}
\setcounter{tocdepth}{2} 
\tableofcontents

\end{titlepage}

\begin{abstract}
We study online temporal voting, where a group of voters submit 0/1 approvals on sets of alternatives that arrive online over multiple rounds and a single alternative is chosen in each round. We introduce online variants of two well-known game theoretic properties, strategyproofness (SP) and independence of irrelevant alternatives. We show that online independence of irrelevant alternatives (OIIA) is a sufficient condition for online strategyproofness (OSP), and that several known online voting rules satisfy OIIA and thus OSP, but that they are not SP. In particular, we show that Perpetual Phragmén, the only known online voting rule to satisfy PJR, satisfies OSP. The Method of Equal Shares (MES), a semi-online voting rule knwon to satisfy wEJR, also satisfies OSP. We then introduce the price of manipulability, which quantifies the effect of strategic behaviour on proportional representation guarantees. Finally, we introduce asymptotic satisfaction of proportional representation and show that an online voting rule, Serial Dictator, is fully strategyproof and satisfies proportional justified representation (PJR) up to an additive constant. 
\end{abstract}

\section{Introduction}
The temporal voting framework studies the problem of satisfying precisely defined notions of individual and group fairness to voters over a sequence of votes. The paradigmatic example is a group of friends that meets regularly for dinner. Each week, they decide on where to eat. If every week they simply go to the restaurant with the most votes, one can easily imagine a scenario where a like-minded, razor-thin majority decides where to eat every week, ignoring the wishes of the minority. Worse still, the dominant group need not even be a majority; if, say, a third of the group are always happy to go to their same favorite pizza restaurant, and nobody else can agree on an alternative even though they all hate pizza, the minority nonetheless compels everyone to eat pizza week after week. The temporal voting problem is the problem of defining fairness in such scenarios and designing voting rules that guarantee this fairness. 

Arguably more impactful applications can be found in the domain of artificial intelligence. Some agentic systems built from an ensemble of LLMs must merge the outputs of the models over sequences of decisions. Proportional representation could help improve performance in a number of settings. Virtual democracy is a setting where machine learning models are tasked with representing the preferences of humans over a large number of votes that would be too time-consuming for a human to participate in themselves. Here, too, the concept of proportional representation over time is a natural and desirable outcome. A detailed discussion can be found in Peters \cite{prop-ai}.

The online setting is natural for temporal voting. We consider scenarios where the alternatives and approvals in each round are revealed sequentially. For example, suppose that the issues at stake in virtual democracy are happening in real-time, then they are revealed sequentially and in this sense the alternatives arrive online. Or suppose that the voters wish to intervene and update their instructions to their representative models as their opinions change and evolve over time, then in this sense the approvals arrive online. Of course, both alternative and approvals can arrive online, and this general setting is the one we consdier. 

Prior work has studied both online and offline settings. The principal fairness axioms considered are notions of proportional representation, that is, ensuring that groups of voters are guaranteed a degree of satisfaction in some sense proportional to their size. Other common game theoretic axioms such as strategyproofness have also been considered. We focus on the intersection of fairness and strategyproofness in the online setting. We show that, although strategyproofness in the general sense remains elusive for temporal voting, there are at least two reasons to be optimistic: (1) a well-motivated relaxation of strategyproofness for the online setting is achievable unconditionally and (2) a strategyproof voting rule can achieve fairness asymptotically.

\subsection{Our contributions}
In section \ref{sec:strategyproofness} we study an online relaxation of strategyproofness for the online temporal voting setting that was introduced in \cite{walsh-fair-division} for fair division of indivisible goods. We show that many known online voting rules satisfy this relaxation, including Method of Equal Shares and Perpetual Phragmén, two voting rules well known for their strong proportionality guarantees. Most known voting rules are known or suspected not to be strategyproof in the general sense, thus highlighting a gap between the online and offline settings.

In section \ref{sec:pom} we introduce analysis of the price of manipulability, where we study what happens to proportionality guarantees when voters submit strategic ballots. Proportional representation guarantees are functions of voter approvals, thus if voters act strategically, it is not clear to what extent proportional fairness guarantees hold. We give lower bounds for JR and PJR when one voter acts strategically.

In section \ref{sec:asymptotic} we consider an asymptotic version of proportional temporal fairness, showing that there exists an online voting rule that gets arbitrarily close to satisfying proportional justified representation (PJR) as the number of voting rounds grows large. We then further asymptotic analysis by considering the price of manipulability as the number of rounds grows large. 

\section{Preliminaries}

\subsection{The model}

We always use $[k] = \{ 1,2,...,k\}$. $N = \{ v_i \}_{i \in [n]}$ is the set of $n$ voters. $T \in \mathbb{N}$ is the number of rounds of voting in an instance. Voters submit approval ballots, where they either approve or not of each available alternative. In round $t$, the set of available alternatives is $C_t$. Voter $v_i$'s approvals in round $t$ are $A_t(v_i) \subseteq C_t$. $A^t(v_i) = \onstratdef{A_t(v_i)}$ is the $t$-tuple of voter $i$'s approvals in rounds $1,2,\dots,t$. The approval profile $A_t = (A_t(v_1), A_t(v_2),\dots, A_t(v_n))$ is the $n$-tuple of approvals of all voters in iteration $t$. $A^t = (A^1, A^2, \dots, A^t)$ is the $t$-tuple of all approval profiles in rounds $1,2,...,t$. Similarly, $C^t = (C_1, C_2,\dots, C_t)$ is the $t$-tuple of sets of alternatives available in rounds $1,2,\dots, t$. An instance is defined by the tuple $(N,A^T,C^T)$.

A voting rule $\vr$ maps (sub)instances $(N,A^T,C^T)$ to a vector of elements $c_t \in C_t$ representing the winning alternative in each round. Thus $\vr(N,A^T,C^T) = \bw \in C_1 \times \dots \times C_t$. Generally, $N$ and $C^T$ are either fixed or clear from the context and are omitted in order to simplify notation, thus we write $\vr(A^T)$. $\vr(A^T)_k$ denotes the winning alternative in round $k$.  

We often consider the union or intersection of the approvals of groups of voters in some particular round. We thus define $\acup_t(G) = \cup_{v\in G}A_t(v)$ and $\acap_t(G) = \cap_{v\in G}A_t(v)$ for $G \subseteq N$. We often consider two different approval profiles that differ only in the approvals of voter $v_i$. $\atti$ is the approval profile with voter $v_i$'s approvals removed in each round, holding all other approvals constant, that is

\[
\atti = ([A_1(v_1),\dots, A_1(v_{i-1}), A_1(v_{i+1}),\dots A_1(v_n)],\dots, [A_T(v_1),\dots, A_T(v_{i-1}), A_T(v_{i+1}),\dots A_T(v_n)])
\]

Thus for some voter $v_i$ with two different approval profiles $A^T(v_i)$ and $\tilde{A}^T(v_i)$, we have $A^T = (\atti, A^T(v_i))$ and $\tilde{A}^T = (\atti, \tilde{A}^T(v_i))$, two approval profiles differing only in voter $v_i$'s approvals. 

Unless stated otherwise, we assume the online setting, and all voting rules are deterministic. This means that the alternative sets and approval sets arrive online, simultaneously\footnote{In an applied setting, the alternatives would arrive first and voters would then submit their approvals. However, in the online setting we give the adversary the power to choose the entire input sequence, thus both alternatives and approvals arrive online simultaneously.}. For this paper we assume complete approvals, that is that all voters have non-empty approvals in all rounds. This assumption rests on the interpretation of empty approvals; if a voter has empty approvals or approves of all alternatives, then this can be interpreted as indifference to the outcome. Although it is a restriction, we believe it is natural and captures a large number or real-life applications. More on this assumption can be found in the model of \cite{welfare_strategy}. In the online setting, a voting rule $\vr$ is a function of the current and all previously seen alternative sets and approval sets, and it makes decisions in each round as though it were the last round. One consequence of this is that any instance $(N,A^T,C^T)$ can be truncated to an instance $(N,A^t,C^t)$ for $t \leq T$ and be a valid instance, and the outcomes for rounds $1,\dots,t$ are identical to the outcomes in the first $t$ rounds of $\vr(A^T)$. 

\subsection{Definitions}
\begin{definition}[Satisfaction]
    The satisfaction of a group of voters over an outcome sequence is the number of rounds in which some voter in the group approves of the outcome. Formally, for group $G\subseteq N$ and approval profile $A^T$, $\sat(G,\vr(A^T)) = |\{ t : \vr(A^T)_t \in \acup_t(G) \}|$. We subscript sat to denote the satisfaction of a group in some round $t$, as in $\sat_t(G,A^T) = 1$ if $\exists\, v \in G$ such that $\vr(A^T)_t \in A_t(v)$, else 0.
\end{definition}

All proportionality axioms are defined in terms of lower bounds on the satisfaction of groups of voters. In order to identify meaningful groups of voters for whom there should be satisfaction lower bounds, we define the following property:

\begin{definition}[$\ell$-cohesiveness]
   A group $G \subseteq N$ is said to be $\ell$-cohesive in $A^T$ if there exists a set of rounds $R \subseteq [T]$ such that $|R| = \ell$ and $\acap_r(G) \neq \emptyset$ for all $r \in R$. Note that the rounds need not be contiguous.
\end{definition}

$\ell$-cohesive groups are simply groups in which all voters agree on some nonempty set of alternatives in some number of rounds. Using this notion, we can define precise proportional representation axioms. In this paper we pay special attention to the justified representation program originating in multi-winner approval voting \cite{aziz}, \cite{pjr} and later adapted to temporal voting \cite{justified_representation}, \cite{chandak}, \cite{strengthening_proportionality}.

\begin{definition}[Justified Representation (JR)]
A voting rule satisfies JR if for all inputs $(N,A^T,C^T)$ and all $G\subseteq N$ that are $\ell$-cohesive in $A^T$, it holds that $\sat(G,\vr(A^T)) \geq \min(1,\flr{\ell \cdot \frac{|G|}{n}})$.
\end{definition}

JR captures the idea that any cohesive group that is sufficiently large or cohesive over sufficiently many rounds should be represented at least once. It is the weakest of the proportionality axioms considered.

\begin{definition}[Proportional Justified Representation (PJR)]
A voting rule satisfies PJR if for all instances $\instance$ and all groups $G\subseteq N$ that are $\ell$-cohesive in $A^T$, it holds that $\sat(G,\vr(A^T)) \geq \flr{\ell\cdot \frac{|G|}{n}}$.
\end{definition}

PJR captures the idea that any group that is cohesive over $\alpha \%$ of rounds and makes up $\beta \%$ of the voters should approve of at least $(\alpha\cdot \beta)\%$ of the outcomes. 

\begin{definition}[Extended Justified Representation (EJR)]
A voting rule satisfies EJR if for all instances $\instance$ and all $G\subseteq N$ that are $\ell$-cohesive in $A^T$, it holds that there exists some voter $v_i \in G$ such that $\sat(v_i, \vr(A^T)) \geq \flr{\ell\cdot \frac{|G|}{n}}$.
\end{definition}

EJR ensures that any $\ell$-cohesive group has at least one voter who has satisfaction greater than or equal to the lower bound. Contrast this with PJR, which has no specific guarantees for individual voters in an $\ell$-cohesive group.

An axiom is said to be in its weak form if it is satisfied for $\ell = T$, i.e. for groups who agree in all rounds, but not necessarily for $\ell < T$. We denote this wJR, wPJR and wEJR. If the axiom is satisfied for any $\ell \leq T$, then it is in its strong form. Clearly, EJR $\implies$ PJR $\implies$ JR, all strong forms imply their weak forms, and wEJR $\implies$ wPJR $\implies$ wJR.

We define the functions $JR(G, A^T)$, $PJR(G, A^T)$ and $EJR(G, A^T)$ as mapping an approval profile sequence $A^T$ and a group $G$ that is $\ell$-cohesive in $A^T$ to the lower satisfaction bound of the respective axiom. For example, $JR(G,A^T) = \min(1, \flr{\ell \cdot \frac{|G|}{n}})$. 

\section{Related work}

Lackner formalizes a framework\footnote{Lackner uses the name perpetual voting; both perpetual and temporal voting appear in the literature, describing essentially the same framework.} \cite{perpetual_voting} which highlights the problem of perpetually ignoring some group of voters in sequential voting.  There is considerable conceptual and technical overlap with the multi-winner approval voting problem \cite{multi_winner}. In multi-winner approval voting, voters submit approval ballots over a set of $m$ candidates, and a voting rule selects $k < m$ winning candidates. Multi-winner approval voting literature has focused much on identifying groups of voters that should qualify for representation and determining the minimal amount of representation to which they are entitled. Aziz et al. \cite{aziz} introduce the JR and EJR axioms for multi-winner approval voting. Fernandez et al. \cite{pjr} introduce PJR. 

Temporal voting is a sequence of rounds of single-winner approval votes where we seek to identify like-minded groups and guarantee representation for them over time, rather than in each round. Bulteau et al. \cite{justified_representation} first provide natural adaptations of JR and PJR for temporal voting. For a detailed survey and analysis of proportional representation axioms, in particular justified representation variants, in temporal voting, we refer the reader to Phillips et al. \cite{strengthening_proportionality}. Lackner adapts a number of voting rules, including Perpetual Phragmén, from multi-winner approval voting to the temporal voting problem in \cite{perpetual_voting}, \cite{axiomatic_lens}, \cite{proportional_decisions}. Chandak et al. \cite{chandak} show that two previously known online voting rules, Perpetual Phragmén and Method of Equal Shares (MES; originally introduced in \cite{mes} for multi-winner approval voting), satisfy PJR and wEJR, respectively (with the caveat that MES is semi-online, meaning that it must know the total number of rounds from the beginning). They also show that an offline algorithm satisfies EJR, and leave the question open as to whether any online voting rule satisfies wEJR. Elkind et al. \cite{verifying_proportionality} partially answer this question showing that that no semi-online (and therefore no online) voting rule satisfies strong EJR. 
 
There are conclusive results on the existence of strategyproof, proportionally representative voting rules in multi-winner approval voting. Peters \cite{multi-sp} provides an interesting computer-assisted proof that no multi-winner approval voting rule satisfies both SP and a proportionality axiom even much weaker than, and implied by, wJR. It would seem reasonable to conjecture that the same is true for temporal voting, but as of this writing there is no proof of this, nor is there a full characterization of strategyproof temporal voting rules. With that said, strong results do exist. Elkind et al. \cite{welfare_strategy} study strategyproofness in the context of welfare maximization. They show that approval voting (AV), the simple voting rule that selects an alternative with maximal approvals each round, is strategyproof and maximizes utilitarian welfare (essentially the sum of voters approving of the outcome over all rounds). In contrast, they prove that any deterministic voting rule that maximizes egalitarian welfare (maximizing the minimal satisfaction of any voter) fails even a relaxation of strategyproofness.
 
Lackner et al. \cite{free-riding} have studied the phenomenon of free-riding in sequential decision-making. Free-riding is when a voter withholds submitting truthful approvals for popular candidates in the hopes of appearing under-represented. This strategy is of course useful when a voting rule has proportionality guarantees: if an approved alternative will win regardless of whether some individual approves of it, a strategic voter can in effect pretend to be unsatisfied until such time as the voting rule must recognize them as an under-represented voter and give them more representation than they are entitled to, given that they have actually already approved a potentially large number of outcomes. Lackner et al. provide two simple conditions that imply the possibility of free-riding, and by implication, manipulability. Although justified representation axioms do not seem to necessarily imply these conditions, in practice, it does seem to be the case that most known voting rules that satisfy justified representation variants do. This supports the conjecture that SP and proportionality are, as in multi-winner approval voting, mutually exclusive.

Kozachinskiy et al. \cite{optimal_bounds} perform a type of asymptotic analysis in that they bound dissatisfaction for individual voters, where dissatisfaction is the number of rounds in which they do not approve the result. As far as we know, no other work has been done on asymptotic analysis in temporal voting. Many asymptotic results exist in computational social choice, for example in Benade et al. \cite{benade}. We note, however, that our analysis does not involve randomization, but is carried out on a deterministic voting rule. 

Finally, Elkind et al. \cite{welfare_strategy} carry out an analysis of price of proportionality. This is the ratio between an outcome optimal in terms of a welfare objective and one optimal in terms of the same welfare objective but under the constraint the the outcome must be proportional. In this case, proportionality is roughly defined as guaranteeing each voter at least $1/n$ of their maximum utility. 
 
\section{The existence of fair and truthful online temporal voting rules} \label{sec:strategyproofness}

Given knowledge of how a voting rule selects winning candidates, a voter may declare false approvals to cause the voting rule to select more candidates of which the voter approves than otherwise would have been selected, had truthful approvals been declared. A voting rule that does not allow for such strategic manipulation is called strategyproof or truthful. Formally,

\begin{definition}[Strategyproofness (SP)]
    A voting rule $\vr$ is SP if for any input instance $\instance$ and any voter $v_i \in N$ with truthful approvals $A^T(v_i)$, for any $\tstrat \neq A^T(v_i)$ it holds that 
\[
\sat(v_i, \vr(\atti, \tstrat)) \leq \sat(v_i, \vr(A^T))
\]
\end{definition}

SP is a strong property that few voting rules of interest satisfy. Simple approval voting (AV), where in each round an alternative with maximum approvals is selected, is know to be SP \cite{welfare_strategy}. Conceptually, any voting rule that is a variation of AV with arbitrary weights on voters would also be SP, as would any number of trivial voting rules (e.g. ignore all approvals and select an arbitrary alternative in each round). However, the canonical objective in temporal voting is to design voting rules that satisfy fairness axioms, which AV and equivalent rules fail to satisfy (the word problem of the friends meeting for dinner shows this). As far as we know, there are no known voting rules, online or otherwise, that satisfy both SP and some non-trivial proportional representation axiom.

In the online setting, we can define a relaxation of SP. A voting rule is online strategyproof if in any round, fixing all previous rounds and all other voter approvals in the current round, no voter can declare false approvals so as to gain greater satisfaction in the current round.

\begin{definition}[Online strategyproofness (OSP)]
A voting rule $\vr$ is OSP if for any input instance $\instance$, any round $t \leq T$ and any voter $v_i$ with truthful approvals $A^t(v_i)$, for any $\strat \neq A^t(v_i)$ where $\strat = \onstratdef{\tilde{A}_t(v_i)}$, it holds that 
\[
\sat_t(v_i, \vr(\ati, \strat)) \leq \sat_t(v_i, \vr(A^t))
\]
\end{definition}

OSP captures the myopic nature of the online setting. If there is much uncertainty about future rounds, or if voters have low inter-temporal substitution preferences (i.e. they mainly only care about what is happening right now), then OSP is well-motivated. We now adapt two familiar properties from social choice and algorithmic game theory to our setting.

\begin{definition}[Monotonicity]
A temporal voting rule $\vr$ is monotone if for all instances $\instance$, all $t \leq T$ and all $v_i \in N$, if $\vr(A^t)_t = c$ and $\onstrat = (A_1(v), A_2(v),...,A_t(v) \cup \{ c \})$ then $\vr(\ati, \onstrat)_t = c$.
\end{definition}

Monotonicity is the property that if for any instance the outcome in round $t$ is $c$, then adding $c$ to the approval set of any voter in that round does not change the outcome. We now show that monotonicity is a necessary condition for OSP.

\begin{lemma}
If an online voting rule is not monotone, then it is not OSP.    
\end{lemma}
\begin{proof}
Suppose a voting rule $\vr$ is not monotone. Then there exist $\instance$, $t \leq T$ and $v_i \in N$ such that $\vr(A^t)_t = c$ but $\vr(\ati, \strat)_t = c'$ for some $c' \neq c$ and $\strat = \onstratdef{A_t(v_i) \cup \{c\}}$. Note that $c \not\in A_t(v_i)$, because if $A_t(v_i) \cup \{ c \} = A_t(v_i)$, then $\vr(A^t)_t = \vr(\ati, \strat)_t$. We now need to show that a strategy exists in round $t$ that gets $v_i$ strictly greater satisfaction, for any truthful approvals. 

We consider two cases, $c' \not\in A_t(v_i)$ and $c' \in A_t(v_i)$. If $c' \not\in A_t(v_i)$, then for voter $v_i$ with true approvals $\strat = \onstratdef{A_t(v_i) \cup \{ c \}}$, $\sat_t(v_i, \vr(A^t)) > \sat_t(v_i, \vr(\ati, \strat))$. Similarly, if $c' \in A_t(v_i)$, then for voter $v_i$ with true approvals $A^t(v_i)$ (remember that $c \not\in A_t(v_i)$), $\sat_t(v_i, \vr(\ati, \strat)) > \sat_t(v_i, \vr(A^t))$. In all cases, we have shown a strategy exists for a voter to derive greater satisfaction in round $t$ by submitting false approvals, therefore $\vr$ is not OSP.
\end{proof}

It remains an open question whether monotonicity is a sufficient condition for OSP. We now adapt independence of irrelevant alternatives to the online temporal voting setting, and show that it is a sufficient condition for OSP.

\begin{definition}[Online independence of irrelevant alternatives (OIIA)]
A voting rule $\vr$ satisfies OIIA if for all $\instance$, $t \leq T$ and $v_i \in N$, if $\vr(A^t)_t = c$ then for $\onstrat = \onstratdef{A_t(v_i) \setminus \{ d \}}$, it holds that $\vr(\ati, \onstrat)_t = c$ for all $d \neq c$. 
\end{definition}

If a voting rule is OIIA, then in any round, if a voter removes a non-winning alternative from its approval set, the result is unchanged. A useful implication of OIIA is that if any voter adds an approval $d$ to their approval set in any iteration $t$, then if the outcome is changed it is changed to $d$. 

\begin{lemma} \label{lemma:oiia}
If a voting rule $\vr$ satisfies OIIA, then for any instance $\instance$, $t \leq T$, voter $v_i \in N$, $d \in C_t$, $d\not\in A_t(v_i)$ and $\strat = \onstratdef{A_t(v_i) \cup \{ d \}}$, if $\vr(A^t)_t \neq \vr(\ati, \strat)_t$ then $\vr(\ati, \strat)_t = d$.
\end{lemma}
\begin{proof}
Suppose a voting rule $\vr$ is OIIA, but that there exists some instance $\instance$ and $t\leq T$ such that $\vr(A^T)_t =c$ but $\vr(\ati, \strat)_t = c' \neq c$ for some $\strat = \onstratdef{A_t(v_i) \cup \{d\}}$, $v_i \in N$ and $d \neq c'$, $d \not\in A_t(v_i)$. Then we have a violation of OIIA, as removing $d$ from $\tilde{A}_t(v_i)$ changes the outcome from $c' \neq d$ to $c \neq c'$.
\end{proof}

\begin{theorem}
    OIIA implies OSP.
\end{theorem}
\begin{proof}
Fix any voting rule $\vr$ that satisfies OIIA, input instance $(N,A^T,C^T)$, $v_i \in N$ and $t \leq T$. $\vr(A^t)_t = c$. We assume that $A^t(v_i)$ is truthful and this is without loss of generality, as it holds for any possible truthful approvals of any voter in any instance. Suppose that $c \in A_t(v_i)$, then $\sat_t(v_i, \vr(A^t)) = 1 \geq \sat_t(v_i, \vr(\ati, \onstrat))$ for all $\ostrat$ and $\onstrat = \onstratdef{\ostrat}$. Assume then that $c \not\in A_t(v_i)$. Suppose there exists some approval set $\ostrat$ such that $\vr(\ati, \strat)_t = c' \neq \vr(A^t)_t$, then there exists a sequence of edits to $A_t(v_i)$ that are either removals from or additions to the set $A_t(v_i)$. Wlog, every edit in the sequence happens only once, i.e. no alternative is added and then later removed (or removed and later added)\footnote{This is without loss of generality because for any sequence in which some alternative $d$ is added followed by other edits and then $d$ is removed, the resulting set is identical to the one resulting from the sequence in which $d$ was never added to being with. The same is clearly true for the case of removing an alternative and adding it back later.}. In any step of the sequence where some alternative $d \not\in A_t(v_i)$ is added, then by lemma \ref{lemma:oiia}, either the outcome does not change or it changes to $d$. In any step of the sequence, if some alternative $d$ is removed, then it must be an alternative $d \in A_t(v_i)$, and the outcome must be unchanged since the original winning outcome was $c \not\in A_t(v_i)$ and every addition $d'$ can only change the outcome to some $d' \not\in A_t(v_i)$. In all cases, the outcome is either unchanged and remains $c \not\in A_t(v_i)$, or it is changed to $d \not\in A_t(v_i)$, resulting in no gain in satisfaction and thus $\sat_t(v_i, \vr(\ati, \strat)) = 0 = \sat_t(v_i, \vr(A^t))$.
\end{proof}

\begin{proposition}
    OIIA is not a necessary condition for OSP.
\end{proposition}
\begin{proof}
Let IrrelevantDictator be the voting rule that chooses the lexicographically first element in voter $v_1$'s approval set every round, except in the case where $|C_t| = 3$, in which case it maps $v_1$'s approvals to the outcome in the following way:

\begin{center}
    $\{a\} \rightarrow \{a\},\ \{b\} \rightarrow \{b\}, \{c\} \rightarrow \{c\}$
    
    $\{ a,b \} \rightarrow \{ a \}, \{ a,c \} \rightarrow \{ a \}, \{ b,c \} \rightarrow \{ b \}$
    
    $\{a,b,c\} \rightarrow \{ c \}$
\end{center}

Clearly, IrrelevantDictator is SP as any voter $v_i$, $i \neq 1$, cannot influence the outcome therefore $\sat(v_i, A^T) = \sat(v_i, (\atti, \tstrat))$ for all $\tstrat$. As for voter 1, for any truthful approvals $A^T(v_1)$, $\sat(v_1, \vr(A^T)) = 1 \geq \sat(v_1, \vr(A^T_{-1}, \tilde{A}^T(v_1)))$ for all $\tilde{A}^T(v_1)$. IrrelevantDictator also clearly violates OIIA, because when $|C_t| = 3$, if voter 1 declares $\{a,b,c\}$, the outcome is $c$, but declaring $\{ a,b,c \} \setminus \{ b \} = \{ a,c \}$ changes the outcome to $a$, a violation of OIIA. 
\end{proof}

Although OIIA may seem like a very strong property, we show that it is natural in that many known voting rules satisfy it. This means that there exist voting rules that are OSP.

\subsection{Weighted Approval Methods (WAMs)}

WAMs are a large class of voting rules intorudced by Lackner \cite{perpetual_voting} including any rule that assigns initial weights to voters, in each round chooses an alternative with maximum score, where the score of an alternative is defined as the sum of the weights of all voters who approve of that alternative, and then updates voter weights based on their approval of the winning alternative.

{
{
\SetAlgorithmName{FRAMEWORK}{Framework}{List of Frameworks}
\begin{algorithm}[H]
\caption{Weighted Approval Method}
$\alpha_i \gets k_i$ for $i \in [n]$ and constants $k_i$

fix $f : \mathbb{R} \rightarrow \mathbb{R}$ such that $f(x) \leq x$

fix $g : \mathbb{R} \rightarrow \mathbb{R}$ such that $g(x) \geq x$

$\score_t(c) := \sum\limits_{i \in [n], c \in A_t(v_i)} \max(0, \alpha_i)$ 

\For{$t = 1,\ t \leq T$}{
$C^*_t = \underset{c}{\argmax} \{ \score_t(c) \}$

$\bw_t \gets c^* \in C^*_t$ \tcp*{breaking ties arbitrarily}

\For{$i \in \{ i : c^* \in A_t(v_i) \}$}{
$\alpha_i \gets f(\alpha_i)$ 
}

\For{$i \in \{ i : c^* \not\in A_t(v_i) \}$}{
$\alpha_i \gets g(\alpha_i)$
}

}

\end{algorithm}
}
}

Clearly, any voting rule in this framework is an online algorithm, as only information in the current and previous rounds is ever used. We now show that all WAMs are OSP.

\begin{proposition}
    All WAMs are OSP.
\end{proposition}

\begin{proof}
We prove this by showing that all WAMs satisfy OIIA. We assume lexicographical tie-breaking\footnote{Lexicographical tie-breaking is not uniquely required. Other deterministic tie-breaking rules that are not functions of the voting rule's past decisions could also work just as well.}. Let $\vr$ be any WAM, and fix any instance $(N,A^T,C^T)$, $t \leq T$ and any voter $v_i \in N$. We observe the effect of removing from $v_i$'s approval set a losing alternative. In round $t$, $\vr(A^t)_t = a$. Let $\tilde{A}^t = (\ati, \onstrat)$ and $\onstrat = \onstratdef{A_t(v_i) \setminus \{ d \}}$ for some $d \in A_t(v_i)$, $d \neq a$. If $\alpha_i \leq 0$, then $v_i$ does not contribute to the score of any alternative and $\vr(\tilde{A}^t) = \vr(A^t)$, thus we assume that $\alpha_i > 0$. Clearly, 

\[
\sum\limits_{j\in[n], c \in A_t(v_j)} \max(0, \alpha_j) = \sum\limits_{j\in[n], c \in A_t(v_j)} \max(0, \alpha_j)
\]

for all $c \neq d$ and 

\[
\sum\limits_{j\in[n], d \in \tilde{A}_t(v_j)} \max(0, \alpha_j) < \sum\limits_{j\in[n], d \in A_t(v_j)} \max(0, \alpha_j)
\]

The only score that changes is that of $d$, which is strictly less than when $v_i$ declared $A^t(v_i)$, so it cannot be minimal and all other alternatives have unchanged scores and must lose tie-breaking as when $A^t(v_i)$ was declared. Therefore $\vr(\tilde{A}^t)_t = \vr(A^t)_t$ and $\vr$ satisfies OIIA.
\end{proof}

\subsection{GreedyJR}

GreedyJR is a WAM with $f(\alpha) = 0$ and $g(\alpha) = \alpha$ that initializes all weights to 1. It is in effect very similar to the GreedyCC algorithm described in \cite{justified_representation}, which differs only slightly because of some assumptions on approvals. GreedyJR can be viewed as a greedy algorithm that picks an alternative with max score in each round, and then removes from consideration in future rounds all voters who approved of that alternative. 

\begin{proposition}
    GreedyJR satisfies JR.
\end{proposition}

\begin{proof}
Suppose towards contradiction that GreedyJR has terminated and there exists an $\ell$-cohesive group $G$ such that $\flr{\ell \cdot \frac{|G|}{n}} \geq 1$ and $\sat(G,\vr(A^T)) = 0$. Let $g = |G|$. $\flr{\ell\cdot \frac{g}{n}} \geq 1$ implies that $\ell \geq \frac{n}{g}$. In each of the $\ell$ rounds in which $G$ is cohesive, GreedyJR selected an alternative with score at least $g$, since it always picks an alternative with max score, which is at least $g$ in these rounds. In each round, only voters not yet satisfied are considered in computing the score of each alternative; this ensures that the $g$ satisfied voters in each of the $\ell$ rounds are distinct from all previously satisfied voters. This means that at least $\ell\cdot g = \frac{n}{g}\cdot g = n$ voters are satisfied over the $\ell$ rounds in which $G$ is cohesive, but $G$ has satisfaction 0, a contradiction. 
\end{proof}

Because GreedyJR is a WAM, it follow immediately that it is OSP.

\begin{proposition}
    GreedyJR is not SP.
\end{proposition}

\begin{proof}
    Let the following table represent the truthful approvals of 5 voters and outcomes of GreedyJR in an instance with $T=2$ and $C_1=C_2=\{c_1,\dots,c_4\}$: 
\begin{table}[H]
  \centering
  \begin{tabular}{lcccccc}
    \toprule
    Round & $v_1$ & $v_2$ & $v_3$ & $v_4$ & $v_5$ & winner \\
    \midrule
    1 & $c_1$ & $c_1$ & $c_1$ & $c_2$ & $c_3$ & $c_1$  \\
    2 & $c_1$ & $c_1$ & $c_2$ & $c_2$ & $c_3$ & $c_3$  \\
    \bottomrule
  \end{tabular}
\end{table}
In round 1, $c_1$ has max score, and is selected by GreedyJR. In round 2, voters 1 to 3 are not considered in the computation of scores. We assume wlog that $c_3$ wins a tie-breaking and so $\sat(v_3,A^T) = 1$. If voter 3 instead falsely declares $\{ c_4 \}$ in round 1, then the outcome remains $c_1$, and in round 2, voter 3's approvals contribute to the scores, and GreedyJR selects $c_2$ leading to $\sat(v_3, \vr(\atti, \tstrat)) = 2 > 1 = \sat(v_3, \vr(A^T))$.
\end{proof}

\subsection{Method of Equal Shares (MES)}

\begin{algorithm}[H]
\caption{Method of Equal Shares}
       $b_i \gets 1$ for all $i \in [n]$

       \For{$t = 1$ to $T$}{
        $C^*_t \gets \underset{c \in C_t}{\argmin} \{ \underset{\rho}{\argmin} \sum\limits_{i \in [n], c \in A_t(v_i)} \min(b_i, \rho) \geq n/T \}$
        
        \If{$C^*_t = \emptyset$}{terminate}
        
        \Else{
         $\bw_t \gets c$ for lexicographically first $c \in C^*_t$ 
         
         $b_i \gets b_i - \min(b_i, \rho_c)$ \tcp*{where $\rho_c$ is minimal $\rho$ for $c$}
         }
         
        }
        
\end{algorithm}

MES is a semi-online voting rule, which means that it requires knowledge of $T$, the total number of rounds there will be. It initializes a fixed price $p=n/T$ and each voter's budget $b_i$ to 1. In each voting round, it selects an alternative $c$ that minimizes $\rho$ in the expression $\sum_{i\in [n],c \in A_t(v_i)} \min(b_i, \rho) \geq p$. Each voter $v_i$ approving of $c$ then has $\min(b_i,\rho_c)$ deducted from their budget, where $\rho_c$ is the minimal $\rho$ required to pay for $c$. An important property of MES is that it is non-exhaustive, meaning that there are input instances that terminate at some iteration $t \leq T$ where there are voters with non-zero budget that approve of available alternatives but cannot afford to pay for them. 

A simple illustrative example is $N=\{v_1,v_2\}$, $T=2$, $p=1$, $C_1=C_2=\{ c_1,c_2 \}$ where both voters approve of $c_1$ in round 1 and are charged $p/2 = 1/2$, and then approve of different singletons $c_1$ and $c_2$ in round 2. With a remaining budget of $1/2$ each, neither of them has enough to pay the price of 1 for the singleton they want, and so MES terminates in round 2 without selecting a winning alternative. 

Different procedures are proposed for dealing with premature termination, for example running a secondary voting rule on the remaining rounds, or simply doing nothing and terminating. We assume termination with no completion procedure for our results. Whether or not there exists a completion procedure that preserves OSP for MES is an open question. 

\begin{proposition}
    MES with no completion procedure satisfies OSP.
\end{proposition}

\begin{proof}
We prove this by showing that MES satisfies OIIA. We assume lexicographical tie-breaking. Let $\vr$ be MES, and fix any instance $\instance$, $t \leq T$ and any voter $v_i \in N$. In round $t$, $\vr(A^t)_t = a$. Let $\rho(c, A^t)$ be the minimal $\rho$ required to pay for $c$ in round $t$. Fixing all other voters, suppose that $v_i$ declares $\onstrat = \onstratdef{A_t(v_i) \setminus \{ d \}}$ for any $d \in A_t(v_i)$, $d \neq a$. Let $\tilde{A}^t = (\ati, \onstrat)$. For any $c \neq d$ we have that 

\[
\underset{\rho}{\argmin} \{ \sum\limits_{j \in [n], c \in \tilde{A}_t(v_j)} \min(b_j, \rho) \geq p \} = \underset{\rho}{\argmin} \{ \sum\limits_{j \in [n], c \in A_t(v_j)} \min(b_j, \rho) \geq p \}
\]

as the approvals for $c$ and budgets are identical in $A^t$ and $\tilde{A}^t$. For $d$ we have 

\[
\underset{\rho}{\argmin} \{ \sum\limits_{j \in [n], d \in \tilde{A}_t(v_j)} \min(b_j, \rho) \geq p \} \geq \underset{\rho}{\argmin} \{ \sum\limits_{j \in [n], d \in A_t(v_j)} \min(b_j, \rho) \geq p \}
\]

as $v_i$'s budget is either 0 or greater than 0. In the former case, $\rho(d, A^t) = \rho(d, \tilde{A}^t)$. In the latter case $\rho(d, A^t) < \rho(d, \tilde{A}^t)$; this is because a constant price $p$ is distributed over strictly less budget, as the total available budget to pay for it is exactly the same budget that was available under $A^t$ minus voter $v_i$'s budget, which is positive. Therefore the only minimal $\rho$ that might change is that of $d$, and it either (i) remains the same, in which case it is either not minimal or loses a tie-breaking, as was the case in $\vr(A^t)$ or (ii) it increases in which case it is not minimal. It must then be that $\vr(\tilde{A}^t)_t = \vr(A^t)_t$ and MES satisfies OIIA.
\end{proof}

When voters are restricted to, or expected to play only strategies that maximize satisfaction in the current round, MES is resilient to strategic manipulation. However, when player strategies are unrestricted, MES is susceptible to the free-rider problem.

\begin{proposition}
    MES with no completion procedure does not satisfy SP.
\end{proposition}

\begin{proof}
Let $N = \{v_1, v_2, v_3\}$, $T=4$ and $C_t = \{ c_1, c_2, c_3 \}$ for all $t \leq T$, and so $p = n/T = 3/4$. The following table shows the truthful approvals, winning alternative and budget of each voter at the end of each round.

\begin{table}[H]
  \centering
  \begin{tabular}{lccccccc}
    \toprule
    Round & $v_1$ & $v_2$ & $v_3$ & winner & $b_1$ & $b_2$ & $b_3$ \\
    \midrule
    1 & $c_1$ & $c_1$ & $c_1$ & $c_1$ & 3/4 & 3/4 & 3/4 \\
    2 & $c_1$ & $c_1$ & $c_1$ & $c_1$ & 2/4 & 2/4 & 2/4 \\
    3 & $c_1$ & $c_1$ & $c_1$ & $c_1$ & 1/4 & 1/4 & 1/4 \\
    4 & $c_1$ & $c_2$ & $c_3$ & $\emptyset$ & 1/4 & 1/4 & 1/4 \\
    \bottomrule
  \end{tabular}
\end{table}

MES terminates in round 4 without choosing a winner, as no minimal $\rho$ exists such that an alternative can be paid for. Consider now the following table where $v_3$ plays a strategy: 

\begin{table}[H]
  \centering
  \begin{tabular}{lccccccc}
    \toprule
    Round & $v_1$ & $v_2$ & $v_3$ & winner & $b_1$ & $b_2$ & $b_3$ \\
    \midrule
    1 & $c_1$ & $c_1$ & $c_1$ & $c_1$ & 3/4 & 3/4 & 3/4 \\
    2 & $c_1$ & $c_1$ & $c_2$ & $c_1$ & 3/8 & 3/8 & 3/4 \\
    3 & $c_1$ & $c_1$ & $c_2$ & $c_1$ & 3/16 & 3/16 & 3/4 \\
    4 & $c_1$ & $c_2$ & $c_3$ & $c_3$ & 3/16 & 3/16 & 0 \\
    \bottomrule
  \end{tabular}
\end{table}

By playing this strategy, $v_3$ gets satisfaction 4, which is strictly greater than 3 when declaring truthfully, therefore MES without a completion procedure is not SP.
\end{proof}

\subsection{Perpetual Phragmén}
\begin{algorithm}[H]
\caption{\textsc{Perpetual Phragmén}}
$\lambda_i \gets 0$ for all $i \in [n]$ \\
    
\For{$t = 1$ to $T$}{
    
$\mathscr{S} = \{ S : S \subseteq N,\ \acap_t(S) \neq \emptyset \}$

$S^* \gets \argmin_{S \in \mathscr{S}} \{ \dfrac{1 + \sum_{i \in [n], v_i \in S}\lambda_i}{|S|} \}$

$\bw_t \gets c \in \acap_t(S^*)$ \tcp*{chosen arbitrarily}

\For{$i \in [n],\ v_i \in S^*$}{
$\lambda_i \gets \dfrac{1 + \sum_{i \in [n], v_i \in S^*}\lambda_i}{|S^*|}$ 
}}
\end{algorithm}

Perpetual Phragmén (PP) is fully online and can be computed in polynomial time \cite{chandak}. It is described as a load-balancing algorithm. Each voter begins with load 0, and in each round a load of 1 is distributed to some group. To select an alternative, the algorithm considers all groups cohesive in that round, that is groups with non-empty approval intersections. The potential load of a group in some round is the sum of the voters' loads plus 1, the additional load for that round. PP selects a cohesive group with minimal average potential load, and selects an alternative from the intersection of that group's approvals. It then sets the load of each voter in the winning group to an even share of the group's potential load, thus distributing the additional load of 1 for that round. Note that a voter's load can never decrease, not all voters in a winning group receive the same additional load, and not all voters who approve of the winning alternative incur additional load.

\begin{proposition}
    Perpetual Phragmen satisfies OIIA
\end{proposition}

\begin{proof}
Fix any $\instance$, $v_i \in N$, $t\leq T$ and let $\vr$ be Perpetual Phragmen and assume lexicographical tie-breaking for all tie-breakings. Let $\strat = \onstratdef{A_t(v_i) \setminus \{ d \}}$ for some $d \neq c$, and $\tilde{A}^t = (\ati, \strat)$. We have that $\vr(A^t) = c$ for some $c \in C_t$. Let $\phi(G) = \frac{1 + \sum_{i \in [n], v_i \in G} \lambda_i}{|G|}$ denote the potential average load of a group, that is the load assigned to each voter in $G$ should they win. 

Let $\ncap_t = \{ G \subseteq N : \acap_t(G) \neq \emptyset \}$ be the set of all groups of voters with nonempty approval intersections in round $t$ of approval profile $A^t$, and $\nncap_t = \{ G \subseteq N : \tilde{A}^{\cap}_t(G) \neq \emptyset \}$ be the set of groups of voters with nonempty approval intersections in round $t$ of approval profile $\tilde{A}^t$ . Clearly, $\nncap_t \subseteq \ncap_t$, as removing an alternative from $v_i$'s approval set can only induce fewer groups with nonempty intersection. We note that $\phi(G)$ is the same in $A^t$ and $\tilde{A}^t$, that is, the individual loads are identical in both sets and therefore so too are their sums. It follows that $\argmin_{G \in \nncap_t} \phi(G) \subseteq \argmin_{G \in \ncap_t}\phi(G)$. Let $G^* = \lexmin(\ncap_t)$ be the lexicographically first group in $\ncap_t$. Either $v_i \in G^*$ or $v_i \not\in G^*$. In the former case, $v_i$ is still in $G^*$ after removing approval for $d$, as the winning alternative $c$ was chosen from $\acap_t(G^*)$ and $d \neq c$. In the latter case, clearly $G^*$ is unaffected and remains the lexfirst set of the argmin set. In both cases, we conclude that $G^* = \lexmin(\ncap_t) = \lexmin(\nncap_t)$, and that $\vr(A^t) = \lexmin(\acap_t(G^*)) = \vr(\tilde{A}^t) = c$, and it follows that Perpetual Phragmen satisfies OIIA.
\end{proof}

It follows immediately that Perpetual Phragmen satisfies OSP.

\begin{proposition}
    PP does not satisfy SP.
\end{proposition}

\begin{proof}
    Let $N=\{v_1, v_2, v_3 \}$, $C_1=C_2 = \{ a,b,c \}$, $T=2$. Suppose all voters approve of $a$ in round 1, and approve of different singletons in round 2. After round 1, all voters are assigned a load of 1/3. In round 2, only one voter is satsifed via some tie-breaking. Wlog assume that $v_3$ is not satisfied in round 2. Then $v_3$ would gain strictly more satisfaction by declaring $\{b\}$ in round 1, as PP would still select $a$ and assign a load of 1/2 each to $v_1$ and $v_2$, and thus $v_3$ would have minimal load in round 2.
\end{proof}

\section{Price of manipulability} \label{sec:pom}

If a voting rule satisfies some proportionality axiom but is manipulable, it may only guarantee proportional representation when all voters declare truthfully. By design, justified representation variants do not guarantee minimal representation to all voters in an $\ell$-cohesive group, as such a requirement would be too difficult to satisfy. Even EJR, which does guarantee that at least one voter in the group a significant amount of satisfaction, can leave some voter in a group with as little as 1 satisfaction, even for very large $\ell$. Since the voting rule determines $\ell$-cohesiveness from declared approval sets, and representation is only guaranteed for $\ell$-cohesive groups, strategic voting may break up groups that would be $\ell$-cohesive under truthful declarations, and the voting rule may then fail to preserve the satisfaction lower bound for the truly $\ell$-cohesive groups.  Even very modest opportunities for strategic voting may be enough to incentivize some voters to declare in such a way as to break away from an $\ell$-cohesive group to which they belong and have a significant impact on the satisfaction of the group. We begin by looking at the effect that a single strategic voter can have on the satisfaction lower bounds of $\ell$-cohesive groups. If a voting rule satisfies JR, it turns out that a single strategic voter cannot cause a violation of JR.
\begin{proposition}
Any voting rule $\vr$ that satisfies JR on truthful inputs also satisfies JR if a single voter acts strategically. 
\end{proposition}
\begin{proof}
Fix any instance $(N, A^T,\ C^T)$ with $A^T$ truthful and voting rule $\vr$ that satisfies JR. Assume that $\vr$ is manipulable and that the instance is non-trivially satisfied, that is that there exists some $\ell$-cohesive group $S$ such that $JR(S, A^T) = 1$ (otherwise JR is vacuously satisfied on that instance). Let $G$ be any such group and let $v_i$ be any voter who has a strategy $\tstrat \neq A^T(v_i)$ such that $\sat(v_i, \vr(\atti, \tstrat) > \sat(v_i, \vr(A^T))$. We consider two cases, $v_i \not\in G$ and $v_i \in G$. For the former, since $v_i$ was not part of the group $G$, $G$ still forms a $\ell$-cohesive group in $(\atti, \tstrat)$, and since $\vr$ guarantees JR, the group retains the same guarantee that it did under $A^T$. Suppose then that $v_i \in G$. In the worst case, $v_i$'s untruthful approvals cause $G$ to no longer be $\ell$-cohesive in $(\atti, \tstrat)$, and $\vr$ does not satisfy $G \setminus  v_i$ at all. However, $v_i$ gained strictly greater approval by declaring $\tstrat$, thus $v_i$ must have satisfaction at least 1, and therefore $\sat(G, (\atti, \tstrat)) \geq 1 \geq JR(G, A^T)$. 
\end{proof}

This is essentially a consequence of the fact that JR only guarantees a satisfaction of 1, and a voter only acts strategically if they can gain at least 1 additional satisfaction over declaring truthfully. We now see that for PJR, a single voter can in principle significantly decrease a group's satisfaction. 

\begin{proposition}
    The price of manipulability for a PJR voting rule when a single voter acts strategically is $\Omega(T)$. 
\end{proposition}

\begin{proof}
Fix a manipulable voting rule $\vr$ that satisfies PJR and any instance $(N,A^T,C^T)$ with truthful $A^T$ for which PJR is non-trivially satisfied. Fix any voter $v_i$ with a strategy $\tstrat$ that gets $v_i$ strictly greater satisfaction than declaring truthfully. For any group $G$ such that $v_i \in G$, $G$ may no longer be $\ell$-cohesive in $(\atti, \tstrat)$. Because $\vr$ satisfies PJR, $G\setminus \{ v_i \}$ still has satisfaction $\flr{\ell \frac{|G| - 1}{n}} = \flr{\ell \frac{|G|}{n} - \ell \frac{1}{n}}$. For any $\ell$, $\ell = qn + r$, therefore $\flr{\ell \frac{|G|}{n} - \ell \frac{1}{n}} = \flr{(qn + r)\frac{|G|}{n} - (qn + r)\frac{1}{n}} = q|G| - q + \flr{\frac{r}{n}(|G| - 1)} \geq q|G| - q$. Because $v_i$ acted strategically, it gained satisfaction of at least 1, therefore $\sat(G, \vr(\atti, \tstrat)) \geq q|G| - q + 1$. 
\end{proof}

Thus we see that in principle a single strategic voter can have a significant impact on PJR satisfaction lower bounds for $\ell$-cohesive groups, in particular when $T$ or $|G|$ are small. An analysis for the price of manipulability for multiple strategic voters likely involves an equilibrium analysis and is a very interesting open question beyond the scope of this paper. Moreover, only a general lower bound can be set in this way, and tight lower bounds will require analysis of specific voting rules. 

\section{Asymptotic analysis} \label{sec:asymptotic}

We now turn to analyzing the asymptotic behaviour of voting rules, that is, what happens to group representation when $T$ grows large and in particular when it grows large relative to $n$. 

\begin{definition}[Asymptotic proportionality]
A voting rule $\vr$ satisfies some proportional representation axiom $PR$ asymptotically if for all input instances $(N,A^T,C^T)$ and all $\ell$-cohesive groups $G$ in $A^T$, we have that
\[
\dfrac{\sat(G, \vr(A^T))}{PR(G,A^T)} \leq \rho + o(PR(G, A^T))
\]
\end{definition}

As the number of rounds $T$ grows large, a voting rule that does not satisfy some particular proportionality axiom may approximate it. We now prove the existence of such a voting rule.

\subsection{A SP voting rule that achieves PJR asymptotically}

\begin{algorithm}[H]
\caption{Temporal Serial Dictator}\label{alg:cap}
$\sigma \gets$ an arbitrary permutation of $[n]$

$\chi \gets$ an arbitrary choice function

\For{$t = 0$ to $T-1$}{

$j \gets \sigma((t \mod n) + 1)$

$\bw_{t+1} \gets \chi(A_{t+1}(v_j))$
}
\end{algorithm}

Temporal Serial Dictator (TSD) is a simple round-robin type algorithm very similar to others seen in computational social choice. Clearly, it is both online and efficiently computable. We first prove that it fails to satisfy even weak JR.

\begin{proposition}
    Temporal Serial Dictator does not satisfy wJR.
\end{proposition}

\begin{proof}
    Let $T = 2$, $n$ be even and $C_1=C_2$ such that $|C_1| > 1$. Wlog, let the permutation on voters be identity. Let $G = \{ v_{n/2}, \dots, v_n \}$ and $A_1(v) = A_2(v) \subset C_1$ for all $v \in G$. Let the approval sets for all other voters not in $G$ in rounds 1 and 2 be any arbitrary subsets of $C_1 \setminus \acup_1(G)$. $JR(G,A^T) = \min(1, \flr{2\cdot \frac{n/2}{2}}) = 1$. In both rounds, an alternative from the approval set of some voter not in $G$ is selected, leaving $G$ with satisfaction 0. Clearly, $\sat(G,A^T) < JR(G, A^T)$, thus TSD does not satisfy wJR.
\end{proof}

However, as $T$ grows large, TSD gets arbitrarily close to satisfying PJR, a much stronger proportionality axiom. 

\begin{proposition}
    TSD satisfies PJR asymptotically.
\end{proposition}

\begin{proof}
TSD ensures that at least one voter is satisfied each round, and that every $n$ rounds, each voter is satisfied at least once. For any $T$, $T = qn + r_1$. It follows that in $T$ rounds, any group has satisfaction at least $q|G|$. In particular, over $T$ rounds, any $\ell$-cohesive group has satisfaction at least $q|G|$. In order to satisfy PJR, all $\ell$-cohesive groups must have satisfaction at least $\flr{\ell \cdot \frac{|G|}{n}}$. For any $\ell$, $\ell = q'n + r_2$, and so $\flr{\ell \cdot |G|/n} = \flr{(q'n + r_2)\cdot |G|/n} = \flr{\frac{q'n|G|}{n} + \frac{r_2|G|}{n}} = q'|G| + \flr{\frac{r_2}{n}|G|} \leq (q'+1)|G|$ where the last inequality holds as $r_2 < n$. Thus $PJR(G,A^T) \leq (q'+1)|G|$, and this expression is maximal when $q'=q$. TSD guarantees $\sat(G,\vr(A^T)) \geq q|G|$, therefore the ratio of satisfaction for any $\ell$-cohesive group achieved by TSD is $\dfrac{q|G|}{(q+1)|G|} = \dfrac{q}{q+1} \rightarrow 1$ as $q \rightarrow \infty$.
\end{proof}

There is a known, efficient algorithm that satisfies PJR unconditionally, namely Perpetual Phragmén. However, PP is not SP in contrast to TSD as we now show.

\begin{proposition}
TSD is SP.
\end{proposition}
\begin{proof}
Fix any instance $(N,A^T, C^T)$ and voter $v_i \in N$. For all rounds $t \leq T$, the outcome is a function only of $A_t(v_i)$ for $i = (t \mod n) + 1$, thus we need only consider the strategies that $v_i$ may play when TSD chooses from its approval set. For any such round, if $v_i$ declares truthful approvals $A_t(v_i)$, then it is guaranteed $\sat(v_i, \vr(A^T)) = 1 \geq \sat(v_i, \vr(\ati, \onstrat))$ for any approvals $\onstrat$. 
\end{proof}

Given the elusiveness of voting rules that satisfy SP and proportionality axioms, it is notable that given enough rounds a naive algorithm such as TSD is strategyproof and gets arbitrarily close to PJR.

\subsection{Asymptotic analysis of the price of manipulability}
We continue our asymptotic analysis by looking at what happens to the price of manipulability when $T$ grows large. We first note that for large $T$, JR is satisfied trivially. 
\begin{proposition}
    For $T\geq n$, any voting rule $\vr$ that satisfies JR also satisfies JR under any strategic manipulation.
\end{proposition}

Clearly, if $T \geq n$, then $JR(v, A^T) \geq 1$ for all voters $v \in N$, and thus any $\ell$-cohesive group in $A^T$ has satisfaction 1 and thus JR is satisfied. The case for PJR is slightly more nuanced but we can still use satisfaction lower bounds on individual voters to get a lower bound for any group. 

\begin{proposition}
    For any instance $\instance$ and voting rule $\vr$ that satisfies PJR, the satisfaction of any $\ell$-cohesive group $G$ is at least $\frac{1}{n}PJR(G,A^T)$ as $T$ grows large. 
\end{proposition}

\begin{proof}
For any voting rule satisfying PJR, we can asymptotically bound the satisfaction of any group in a simple way. Fix any $\instance$ and let $\vr$ be a voting rule that satisfies PJR. Let $\tilde{A}^T$ be any strategic approval profile for $N,\ C^T$, that is that some voters submit strategic approvals $\tilde{A}^T(v) \neq A^T(v)$ such that $\sat(v, \tilde{A}^T) > \sat(v, A^T)$. In the worst case, $\vr$ provides exactly minimum satisfaction for all groups, and every $\ell$-cohesive group in $A^T$ is no longer $\ell$-cohesive in $\tilde{A}^T$. However, $\vr$ must still satisfy individual voters as they constitute an $\ell$-cohesive group of size $1/n$. For any $T$, $T = qn + r$ and for any voter $v$, $PJR(v,A^T) = \flr{T  \cdot \frac{1}{n}} = q$, therefore $\sat(v,\vr(A^T)) \geq q$. It follows that any group must have satisfaction at least $q$, as $\sat(G,\tilde{A}^T) \geq \sat(v,\tilde{A}^T)$ for any $v \in G$. It may be that all voters in the group are satisfied in the same rounds, in which case $\sat(G,\tilde{A}^T) = \sat(v,\tilde{A}^T) = q$. For any $\ell$-cohesive group, $PJR(G, A^T) = \flr{\ell \cdot \frac{|G|}{n}} \leq (q+1)|G|$. This gives us the worst-case ratio $\frac{q}{(q + 1)}|G| \rightarrow \frac{1}{|G|} \geq 1/n$ as $q$ grows large. 
\end{proof}

Contrast this loose lower bound with TSD: we see that TSD satisfies PJR up to an additive constant which disappears as $T$ grows large, which is significanlty better than a factor of $1/n$. We also notice that for $n \ll T$, this loose bound may amount to nothing more than a small scaling factor. Moreover, it is not immediately clear whether any voting rules allow for such a high price of manipulability. In any case, this serves as a benchmark against which we can compare the price of manipulability of voting rules that satisfy PJR asymptotically.

\section{Conclusion and open problems}

We have highlighted two significant gaps. The first is the gap between OSP and SP, the online and general concepts of strategyproofness. In the general case, achieving both strategyproofness and proportionality seems difficult, and perhaps impossible. However, in the online setting, where voting rules and agents are assumed to act myopically, many voting rules turn out to be (online) strategyproof. Moreover, it has been shown that some of them do satisfy strong proportionality axioms, namely that GreedyJR satisfies JR, MES satisfies wEJR \cite{chandak} and Perpetual Phragmén satisfies PJR. A significant open question is whether there exists a (semi-)online voting rule that satisfies wEJR; if such a rule exists, does it satisfy OSP?
 
The second gap highlighted is one between multi-winner approval voting and temporal voting. We have shown that in temporal voting, when $T \gg n$, justified representation variants can be approximated (or even trivially satisfied, as is the case for JR). There doesn't seem to be any such analogue for multi-winner approval voting; if, for example, the number of alternatives is much greater than the committee size, or much greater than the number of voters, there is no obvious way in which this makes it easier to satisfy justified representation variants. A natural follow-up question is to ask whether multi-winner temporal voting \cite{temporal_multiwinner}, where in each round multiple winners are chosen, benefits from the same asymptotic leeway as when a single alternative is chosen in each round.

Another open question is whether Temporal Serial Dictator somehow characterizes SP voting rules that asymptotically achieve justified representation variants, analogously to the case of how the Like mechanism characterizes mechanisms that are both SP and EF \textit{ex ante} in online fair division \cite{walsh-fair-division}. If this turns out to be true, it is interesting in that a naive algorithm is best in class. If not, then other interesting (O)SP voting rules may exist that are asymptotically proportional and satisfy other desirable properties. 

Further interesting open questions include the rate of convergence of asymptotic proportionality. Might it be the case that some (O)SP voting rules converge to PJR faster than others? There may also exist a meaningful notion of partial convergence, where it can be shown that within some number of rounds, the satisfaction of `most' groups has converged to at least its PJR lower bound guarantee. 

\printbibliography

\end{document}